\documentclass{article}

\usepackage{authblk} 
\usepackage{balance} 
\usepackage[numbers]{natbib} 
\usepackage{hyperref}
\usepackage{todonotes}
\usepackage{pifont}
\usepackage{amsmath,amsfonts,amssymb,float,color}
\usetikzlibrary{decorations.pathreplacing,arrows.meta,positioning}
\usepackage{amsthm}

\title{Selfish routing games with priority lanes}

\author[1]{Yang Li}
\author[2]{Alexander Skopalik}
\author[2]{Marc Uetz}

\affil[1]{School of Mathematics and Statistics, Northwestern Polytechnical University, Xi'an, Shaanxi 710072, China}
\affil[2]{Department of Applied Mathematics, University of Twente, Enschede 7500, AE, The Netherlands}

\newtheorem{example}{Example}

\newtheorem{observation}{Observation}
\newtheorem{corollary}{Corollary}
\newtheorem{theorem}{Theorem}
\newtheorem{definition}{Definition}

\newtheorem{lemma}{Lemma}

\newcommand{\BibTeX}{\rm B\kern-.05em{\sc i\kern-.025em b}\kern-.08em\TeX}

\begin{document}
\maketitle 
\begin{abstract}

We study selfish routing games  where users can choose between regular and priority service for each network edge on their chosen path.  Priority users pay an additional fee, but in turn they may travel the edge prior to non-priority users, hence experiencing potentially less congestion. For this model, we establish existence of equilibria for linear latency functions and prove uniqueness of edge latencies, despite potentially different strategic choices in equilibrium. Our main contribution demonstrates that marginal cost pricing achieves system optimality: When priority fees equal marginal externality costs, the equilibrium flow coincides with the socially optimal flow, hence the price of anarchy equals $1$. This voluntary priority mechanism therefore provides an incentive-compatible alternative to mandatory congestion pricing, whilst achieving the same result. We also discuss the limitations of a uniform pricing scheme for the priority option. 
\end{abstract}



\section{Introduction}\label{se1}

Managing congestion in modern transportation and communication networks increasingly relies on service differentiation rather than uniform tolling. Instead of imposing mandatory charges, many real-world systems introduce voluntary priority options such as express toll lanes on highways, priority security or boarding at airports, and premium bandwidth on digital platforms. These schemes offer users a choice between a basic or free regular service and a paid priority service that promises lower latency. Crucially, granting priority access creates negative externalities for non-paying users by reducing capacity or increasing delays. Motivated by these practices, we propose a two-tier selfish routing model in which each edge of the network is replicated into a regular lane and a priority lane, and users strategically self-select between them according to their latency-price trade-offs. In this framework, equilibrium flows arise endogenously from voluntary participation rather than compulsory tolls, enabling us to analyze when and how priority-based service differentiation can improve overall efficiency.

Classic results in selfish routing already highlight the inefficiencies stemming from individually rational behavior. 
About a century ago, Pigou \cite{pigou1920economics}  demonstrated that even in the simplest network with two parallel links, the equilibrium attained by selfish (nonatomic) users can exhibit higher total travel time than the system-optimal flow. Later, Braess' paradox \cite{braess1968paradoxon} showed that adding capacity to a network may paradoxically increase equilibrium travel times. These insights motivated the formalization of the Price of Anarchy (PoA), which quantifies the welfare loss due to selfish routing \cite{koutsoupias1999worst,roughgarden2002howbad,roughgarden2005selfish}. A high PoA indicates that user equilibria can significantly degrade network performance compared to the social optimum, underscoring the importance of mechanism design in routing games.

Another fundamental structural question in selfish routing concerns the existence and uniqueness of equilibria.
The foundational work of Beckmann, McGuire, and Winsten \cite{beckmann1956} showed that equilibria arise as minimizers of a convex potential function in nonatomic congestion games with continuous and nondecreasing latency functions, while Dafermos and Sparrow \cite{dafermos1969} established the variational inequality (VI) formulation of the traffic assignment problem. Smith \cite{smith1979} further analyzed conditions for existence, uniqueness, and stability, and Dafermos \cite{dafermos1980} unified equilibrium theory within the VI framework. Moreover, Bernstein and Smith \cite{bernstein1994} proved existence under lower semicontinuous and possibly discontinuous latencies, while De Palma and Nesterov \cite{de1998optimization} developed optimization-based formulations. In the algorithmic game theory community, Roughgarden and Tardos \cite{roughgarden2002howbad} reaffirmed that equilibria always exist under continuous, monotone latencies, and specifically showed that the price of anarchy is at most $4/3$ for linear latency functions.

Despite this rich literature, the existence of equilibria in the priority model as proposed in this paper does not follow directly: while the potential function associated with total flows remains convex, the induced strategy splits may be inconsistent across paths, which motivates another approach taken in this paper, namely using  a fixed-point argument.

A second influential line of work investigates how taxation can improve efficiency. A canonical remedy is marginal-cost (Pigouvian) pricing, which charges every user the congestion externality they inflict on other users, and thus restores optimality, e.g., \cite{pigou1920economics,DafermosTolls1973}. However, implementing Pigouvian taxes in heterogeneous populations is often difficult in practice. The required tolls can be complex, high, and politically unpopular \cite{cole2003pricing,gu2018congestion}. To address this, Wang et al.\cite{wang2016improved} proposed a step-function toll scheme that uses a few congestion thresholds to set discrete prices. More generally, edge-specific pricing tailored to different user types can reach the social optimum under suitable conditions \cite{cole2003pricing}. Similar pricing ideas have also been applied successfully in three-tier edge computing resource markets \cite{li2022pricing}. Harks, Schäfer, and Sieg \cite{harks2008} showed that tolls inducing a target flow always exist and provided a compact polyhedral characterization of feasible tolls. Note, however, that all these works  rely on compulsory payments. In practice, even small congestion fees often face low public acceptance \cite{gu2018congestion}.

Against this background, voluntary priority options offer a feasible middle ground between laissez-faire and strict regulation. By allowing users to pay for faster service while leaving standard access free, operators create a self-selection mechanism that both generates revenue and potentially improves efficiency. However, this also raises new questions: Can priority-based differentiation achieve socially optimal outcomes? Under what pricing rules does it improve upon unpriced routing? And what are the limits of simple, uniform pricing rules compared to edge-specific pricing? This paper gives answers to these questions.\\

\noindent\textbf{Main Contributions}\\

This paper develops and analyzes a rigorous model of selfish routing with priority lanes. Our key contributions are as follows:

\begin{itemize} 
    \item We introduce a new network flow model in which every edge offers users a paid priority option and a regular option, so that priority users impose additional delay for all regular users. For linear latency functions we prove, via a Kakutani correspondence, that an equilibrium always exists.
    \item While an edge’s split between priority and regular flow may be non-unique, we prove that the edge latencies are unique at equilibrium. 
    This is shown by using and adapting the classical approach via variational inequalities. 
    \item We study both edge-specific and uniform pricing schemes: We prove that under edge-specific marginal-cost pricing of the priority option, every equilibrium coincides with the optimum flow of the original instance, yielding a prive of anarchy equal to 1. In contrast, uniform pricing cannot, in general, induce optimal flows: We construct a family of Pigou-type networks in which the best uniform pricing still has a PoA of 4/3 in the worst case.
\end{itemize}

Overall, this work therefore provides a theoretical foundation for understanding tiered service networks in transportation and communication,  
offering new insights for transportation planners and network operators. Our main insight is that priority lanes, while obviously weaker  than taxation, still allow to achieve socially optimal outcomes purely based on voluntary participation. 


\section{The selfish routing model with priority lanes}\label{se2}

For a standard non-atomic network routing game on a directed graph $G=(V,E)$ with vertex set $V$ and edge set $E$, there are $k$ commodities which are characterized by source-destination vertex pairs $(s_{1},t_{1}),(s_{2},t_{2}),\dots, (s_{k},t_{k})$ and demands denoted by finite and non-negative rates $r_1, \dots, r_k$. Let $\mathcal{P}_{i}$ denote the set of $s_{i}$-$t_{i}$ paths and $\mathcal{P} = \cup_{i}\mathcal{P}_{i}$. A feasible flow is a function $f : \mathcal{P}\to \mathbb{R}^+$ so that $\sum_{P\in\mathcal{P}_{i}}f_{P}=r_{i}$ holds for all $i\in[k]$. For a fixed flow $f=(f)_{P}$, we have $f_{e}=\sum_{P: e\in P}f_{P}$ as the flow on each edge $e$. Moreover, there is a load-dependent latency function $\hat{\ell}_{e}$ associated with edge $e\in E$ and $\hat{\ell}_{e}$ is nonnegative and nondecreasing. In this paper, we focus on the linear case, where $\hat{\ell}_e(x) = a_e x + b_e$ with $a_e, b_e \ge 0$. In deviating slightly from the standard notation in selfish routing literature, we assume the effective latency (or cost) of a user on  edge $e$ with flow $f_e$ is the average with respect to $\hat{\ell}_e$, that is,
\begin{equation*}
    \ell_{e}(f_e) :=\frac{1}{f_e}\int_{0}^{f_e}\hat{\ell}_{e}(z)dz\,.\footnote{We naturally define $ \ell_{e}(0)=\hat{\ell}_e(0)$.}
\end{equation*}

We choose this model because later, we introduce the idea of priorities among two types of users, and it yields a cleaner model for how these affect each other. It is also consistent with the idea that users use an edge subsequently and we model the average latency. For a given flow $f$, the latency of a path $P$ is given by  $\ell_{P}(f)=\sum_{e\in P}\ell_{e}(f_{e})$ as usual. The cost $C(f)$ of a flow $f$ in $G$ is the total latency, i.e.,  
\begin{equation*}
    C(f)=\sum_{P\in\mathcal{P}}\ell_P(f)f_P=\sum_{e\in E}\ell_e(f_e)f_e
\end{equation*}
Then the triple $(G,r,\ell)$ is an instance of a network routing problem. A feasible flow minimizing $C(f)$ is the (socially) optimal latency flow.

We extend a standard routing game $(G,r,\ell)$ towards a routing game with priorities $(G,r,\tilde{\ell},\omega)$ as follows. Here, users have two options at each edge $e\in {E}$ of a network ${G}$, namely to use $e$ as a regular ($R$) or priority ($V$) user. Users who choose $V$ can pass through the edge $e$ with lower latency but must pay a priority cost $\omega_{e} \in \mathbb{R}^+$. Additionally, individuals using the priority lane will increase the latency cost of regular users. Intuitively, priority users pay to pass through the edge first, while regular users have to wait for priority users before they can pass. 
For ease of notation, we extend the set of edges $E$ by replacing each edge $e\in E$ by two edges $e^R$ and $e^V$ to obtain the {\em extended edge set} $\tilde{E}$. This allows us to employ the use of the corresponding set of  {\em extended paths} $\tilde{P}$. The users' strategic choices are then captured by  {\em extended flows}  $\tilde{f} : \mathcal{\tilde{P}}\to \mathbb{R}^+$, where $\tilde{f}_{e^V}$ denotes the flow of priority users passing through the priority lane of edge $e$, while $\tilde{f}_{e^R}$ denotes the flow of regular users through $e$. The total flow $f^t$ of an extended flow $\tilde{f}$ is then defined by $f^t_e=\tilde{f}_{e^R}+\tilde{f}_{e^V}$.

As before, instead of tracking the latency for each user, we average the edge cost over the cumulative flow per type, while all regular users are additionally delayed by the set of all priority users.
That means that the average perceived costs for priority users include the latency cost caused by other priority users plus the priority fee, but they are not affected by regular users. The average perceived costs for regular users, however, are affected by both types of users, resulting in the following perceived cost function per type of users:
\begin{eqnarray*}
    \tilde{\ell}_{e^V}(\tilde{f}_{e^V}):=\frac{1}{\tilde{f}_{e^V}}\int_{0}^{\tilde{f}_{e^V}}\hat{\ell}_{e}(z)dz+\omega_{e}\,,\\
    \tilde{\ell}_{e^R}(\tilde{f}_{e^R},\tilde{f}_{e^V}):=\frac{1}{\tilde{f}_{e^R}}\int_{\tilde{f}_{e^V}}^{\tilde{f}_{e^V}+\tilde{f}_{e^R}}\hat{\ell}_{e}(z)dz\,.
\end{eqnarray*}
Again, for flow values equal 0 we define $\tilde{\ell}_{e^R}(0,\tilde{f}_{e^V}):=\hat{\ell}_e(\tilde{f}_{e^V})$, and  $\tilde{\ell}_{e^V}(0):=\hat{\ell}_e(0)+\omega_e$.

Since we  consider linear latency functions in this work, i.e., $\hat{\ell}_{e}(x)=a_{e}x+b_{e}$, $a_{e}, b_{e} \ge 0$, the perceive cost function of priority users is
\begin{eqnarray} \label{eq:LV}
    \tilde{\ell}_{e^V}(\tilde{f}_{e^V})=\frac{1}{2}a_{e}\tilde{f}_{e^V}+b_{e}+\omega_{e}\,,
\end{eqnarray}
and the average perceive cost of regular users is 
\begin{eqnarray} \label{eq:LR}
    \tilde{\ell}_{e^R}(\tilde{f}_{e^R},\tilde{f}_{e^V})=a_{e}\tilde{f}_{e^V}+\frac{1}{2}a_{e}\tilde{f}_{e^R}+b_{e}\,.
\end{eqnarray}
The perceived cost function $\tilde{\ell}_{\tilde{P}}(\tilde{f})$ of an extended path $\tilde{P} \in \mathcal{\tilde{P}}$ is naturally defined as the sum of the costs of its corresponding edges from $\tilde{P}$.

\begin{definition} \label{def:ns}
An extended flow $\tilde{f}$ for $(G,r,\tilde{\ell},\omega)$ is an equilibrium flow if for all $\tilde{P}_1, \tilde{P}_2 \in \tilde{\mathcal{P}}_i$ with $\tilde{f}_{\tilde{P}_1}>0$, 
\begin{equation*} 
    \tilde{\ell}_{\tilde{P}_1}(\tilde{f})\le \tilde{\ell}_{\tilde{P}_2}(\tilde{f}).
\end{equation*}
\end{definition}

Before turning to the general analysis, we illustrate the perceived cost behavior with a simple single-edge example.

\begin{example}\label{ex:1}
    

\begin{figure}[t]
    \centering
    \includegraphics[width=0.8\linewidth]{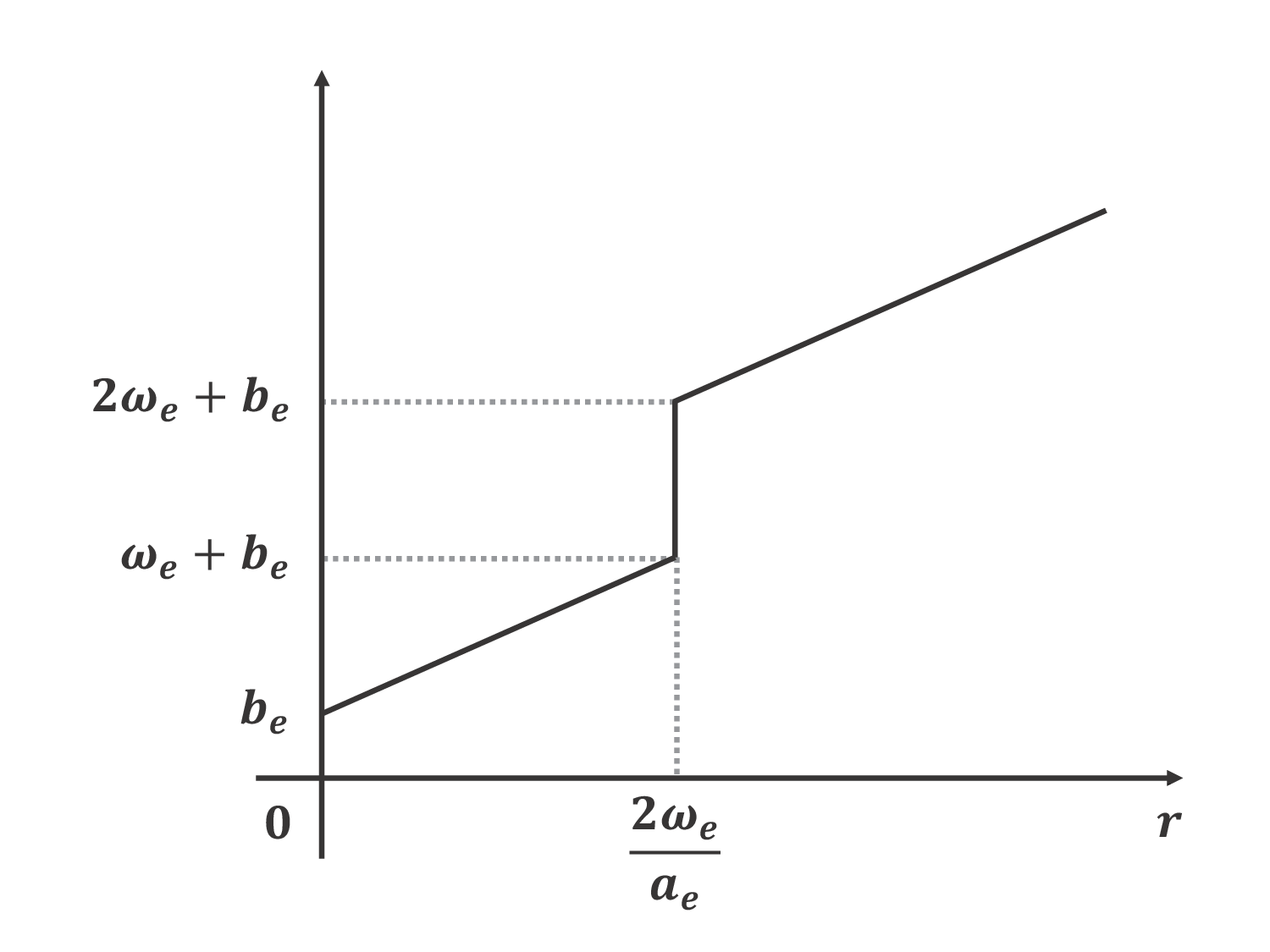}
    \caption{The perceived cost in equilibrium as a function of the total flow $r$ on  the single edge $e$ of Example \ref{ex:1}.}
    \label{fig:lf}
\end{figure}

We consider a network consisting of just a single edge $e$, a demand of $r$, and priority cost of $\omega_e$. The equilibrium depending on $\omega_e$ can be characterized as follows:
\begin{enumerate}
    \item If $r<\frac{2\omega_{e}}{a_{e}}$, all users use the regular lane as their perceive cost is  $$\tilde{\ell}_{e^R}(r,0)=\frac{1}{2}a_{e}r+b_{e} < b_e + \omega_{e} = \tilde{\ell}_{e^V}(0)$$ since $\omega_{e} > \frac12 a_e r$
    \item If $r>\frac{2\omega_{e}}{a_{e}}$, all users use the priority lane as 
   their perceived cost is $$\tilde{\ell}_{e^V}(r)=\frac{1}{2}a_{e}r+b_{e}+\omega_{e} < \tilde{\ell}_{e^R}(0,r)=a_{e}r+b_{e}$$ since $\omega_{e} < \frac12 a_e r$.
   \item  If $r=\frac{2\omega_{e}}{a_{e}}$, the users are indifferent between the two options, hence  any $\tilde{f}_{e^V} \in[0,r]$ and $\tilde{f}_{e^R} = r - \tilde{f}_{e^V}$ constitutes and equilibrium as 
   $$\tilde{\ell}_{e^V}(\tilde{f}_{e^V})=   \tilde{\ell}_{e^R}(\tilde{f}_{e^R},\tilde{f}_{e^V}).$$ 
\end{enumerate}

\end{example}

At equilibrium, we obtain the cost as a function of the total flow on edge $e$, as illustrated in Fig. \ref{fig:lf}. Interestingly, the equilibrium cost at flow value $2\omega_e/a_e$ therefore adopts any value in the range $[\omega_e+b_e, 2\omega_e+b_e]$.  Based on this, we formulate the following observation which will become useful later.
\begin{observation}
\label{observation1}
        If the total flow on an edge $e$ is $\frac{2\omega_{e}}{a_{e}}$  then for any choice of users to distribute over $\tilde{f}_{e^V}$ and $\tilde{f}_{e^R}$,  we have that  $\tilde{\ell}_{e^V}(\tilde{f}_{e^V})=   \tilde{\ell}_{e^R}(\tilde{f}_{e^R},\tilde{f}_{e^V}) \in [\omega_e+b_e,2\omega_e+b_e]$. 
\end{observation}


\subsection{Existence of equilibria.} 

Before establishing our first equilibrium existence result, we briefly discuss our methodological choice. In classical nonatomic routing games, the existence of a Wardrop equilibrium is typically shown via the minimization of a potential function over feasible flows. In our setting, although the potential function associated with the total edge flows remains convex and strictly increasing, 
it is not clear how 
this approach could be adapted to establish equilibrium existence. The reason is that, even though we can guarantee the existence of a total flow vector that minimizes the potential function, it is not straightforward how this solution can be used to define an equilibrium in the extended instance. Specifically, the complication is that the implied strategy distributions on some edges may turn out to be inconsistent. That is, different paths may require different splits between priority and regular users on one and the same edge.
This ambiguity motivates another approach 
to establish the existence of an equilibrium, namely via Kakutani’s fixed-point theorem which simplifies a similar proof of \cite{smith1979} as it is applied directly to our games at hand.


\begin{theorem} \label{theo:1} 
Given a price vector $\omega$, an instance $({G},{r},\tilde{\ell},\omega)$ 
admits at least one equilibrium flow. 
\end{theorem}

\begin{proof}
We proof the theorem by a reduction to Kakutani's fixed point theorem~\cite{Glicksberg52,Fan52}.
The set of all feasible flows $\tilde{\mathcal F}$ of the extended network is a closed, bounded (hence compact) polytope in $\mathrm{R}^{|\tilde{\mathcal{P}}|} $, and obviously convex. 
\begin{equation*}
    \tilde{\mathcal F }
       = \Bigl\{\,\tilde{f}\in \mathrm{R}_+^{|\tilde{\mathcal{P}}|} :  
             \sum_{\tilde{P}\in \tilde{\mathcal{P}}_i}\tilde{f}_{\tilde{P}}=r_i , \forall i \in[k] \Bigr\}.
\end{equation*}


We define the following set-valued mapping $B$ that maps a flow $\tilde{f}$ to the set of feasible flows that only use paths that have minimal perceived cost in $\tilde{f}$.

\begin{equation*}
    B(\tilde{f})\;=\;\prod_{i=1}^k B_i(\tilde{f}),
\end{equation*}
where
\begin{align*}
B_i(\tilde{f})&=B_i((\tilde{f_i},\tilde{f_{-i}})) 
   = 
   \Bigl\{\,\tilde{f}'_i\in\mathrm{R}_{\ge 0}^{|\tilde{\mathcal{P}}_i|}:
      \sum_{\tilde{P}\in \tilde{\mathcal{P}}_i}\tilde{f}'_{\tilde{P}}=r_i,\; \text{with} \;
     \\ &\tilde{f}'_{\tilde{P}}>0\;\Rightarrow\; 
      \tilde{\ell}_{\tilde{P}}(\tilde{f})\;=\;\min_{\tilde{Q} \in \tilde{\mathcal{P}}_i}\tilde{\ell}_{\tilde{Q}}(\tilde{f}) \Bigr\}.
\end{align*}

That defined, notice that if $\tilde{f}\in B(\tilde{f})$, that means for all $i$, all used ($s_i,t_i$)-paths are minimum cost paths. In other words, $\tilde{f}$ is an equilibrium.

We next check that $B\colon \tilde{\mathcal F}\rightarrow 2^{\tilde{\mathcal F}}$ is a Kakutani correspondence:
\begin{enumerate}
    \item Since $\mathcal{P}_i$ is finite, there is at least one path of minimal cost, $B_i(\tilde{f})\neq\emptyset$.
    \item If two different redistributions $\tilde{f}'_i,\,\tilde{f}''_i\in B_i(\tilde{f})$ assign positive mass only to minimal paths, then any convex combination $\alpha \tilde{f}'_i+(1-\alpha)\tilde{f}''_i$ also only puts positive mass on those same minimal paths, so lies in $B_i(\tilde{f})$.  Thus each $B_i(\tilde{f})$ is convex, and so is the product $B(\tilde{f})$.
    \item 
    Based on the continuity of $\tilde{\ell}_{\tilde{P}}(\tilde{f})$ as per eqs.\  \eqref{eq:LV} and \eqref{eq:LR},  for all paths $\tilde{P}$ in the extended network, we have for the perceived costs that $\tilde{\ell}_{\tilde{P}}(\tilde{f}^n)\to \tilde{\ell}_{\tilde{P}}(\tilde{f})$ whenever $\tilde{f}^n\to\tilde{f}$. So if  $ (\tilde{f}^n,\tilde{f}'^n) \to (\tilde{f},\tilde{f}')$ so that each $\tilde{f}'^n\in B(\tilde{f}^n)$, then by continuity of the cost function also $\tilde{f}'\in B(\tilde{f})$, since $\tilde{f}'$ can place positive mass only on paths whose cost equals the minimum cost at $\tilde{f}$. Hence $(\tilde{f},\tilde{f}')$ lies in the graph of $B$, which is therefore closed. 
\end{enumerate}

Because $\tilde{\mathcal F}$ is nonempty, compact and convex, and $B$ satisfies all the conditions of Kakutani’s fixed point theorem, there exists an extended flow $\tilde{f}\in \tilde{\mathcal F}$ with
    $\tilde{f}\in B(\tilde{f})$.
\end{proof}

\subsection{Discussion on equilibrium uniqueness}

In classical selfish routing games (with strictly increasing latency functions), the equilibrium is known to be unique in terms of the total flow on each edge. 
A natural question is whether such uniqueness still holds in our priority model with regular and priority users. The technical challenge here is that the cost function in terms of total flow value may not be continuous due to a discontinuity when the flow is at the threshold $2\omega_e/a_e$, as shown in Figure~\ref{fig:lf}. Nevertheless, we next show that 
the equilibrium remains unique in terms of the latency 
on each edge, even if the distribution between regular and priority strategies may not be. 
We will argue that the total edge flows still fulfill the well known variational inequality
that is known to characterize equilibria in classical network routing games~\cite{dafermos1980}. Because the cost functions in terms of total edge flows are not a function but a correspondence as per Observation~\ref{observation1}, we next give a proof that these variational inequalities still hold. For this, we construct a piecewise linear, single-valued function, to be able to connect the perceived costs in the extended model to the latencies. 

For an extended flow $\tilde{f}$ with a total flow $f^t$ we defines a a single-valued, non-decreasing function $g_e(\cdot)$.
For each edge $e\in E$, with $a_e=0$, we define $g_e(x)=b_e$. For each edge $e$ with $a_e > 0$, we let
\begin{equation*}
g_{e}(x)=\left\{
\begin{array}{lcl} 
\frac{1}{2}a_{e}x+b_{e}, & & {x <\frac{2\omega_{e}}{a_{e}};} \\
\xi_e, & & {x=\frac{2\omega_{e}}{a_{e}} \text{ and } f^t = \frac{2\omega_{e}}{a_{e}}; } \\
\frac{1}{2}a_{e}x+b_{e}+\omega_{e}, & & {x=\frac{2\omega_{e}}{a_{e}}  \text{ and }f^t   \ne \frac{2\omega_{e}}{a_{e}} ;} \\
\frac{1}{2}a_{e}x+b_{e}+\omega_{e}, & & {x >\frac{2\omega_{e}}{a_{e}} ;}
\end{array}\right.
\end{equation*}
where $\xi_e$ denotes the equilibrium edge cost at $e$ in $\tilde f$, i.e.,
\begin{equation*}
    \xi_e := \min\bigl\{ \tilde{\ell}_{e^V}(\tilde{f}_{e^V}),\;\tilde{\ell}_{e^R}(\tilde {f}_{e^R},\tilde{f}_{e^V}) \bigr\}.
\end{equation*}


\begin{lemma}[Variational inequality characterization]\label{lem:vi}
Let $\tilde f$ be an equilibrium of the extended instance $(G,r,\tilde\ell,\omega)$, and  $f^t$ its total flow. 
Let $g(f)=(g_e(f_e))_{e\in E}$. 
Then the following variational inequality holds for all feasible total flows $f$:
\begin{equation}\label{eq:vi}
\langle g(f^t),\, f - f^t \rangle \;\ge\; 0 
\end{equation}
\end{lemma}

\begin{proof}
For each commodity $i$, let
\begin{equation*}
    \lambda_i := \min_{P\in\mathcal{P}_i} \sum_{e\in P}\xi_e
\end{equation*}
be the shortest-path cost at equilibrium. Because $\tilde{f}$ is an equilibrium, and since $g_e(f^t_e) =\min\bigl\{ \tilde{\ell}_{e^V}(\tilde{f}_{e^V}),\;\tilde{\ell}_{e^R}(\tilde {f}_{e^R},\tilde{f}_{e^V}) \bigr\}$, the path flows of $f^t$ satisfy 
\begin{equation*}
\sum_{e\in P}\xi_e = \lambda_i \quad \text{if } f^t_P>0, 
\qquad 
\sum_{e\in P}\xi_e \ge \lambda_i \quad \text{if } f^t_P=0.
\end{equation*}
Now observe that
\begin{equation*}
\begin{aligned}
\langle g(f^t), f-f^t\rangle
&= \sum_{e\in E} g_e(f^t_e)\,(f_e-f^t_e) \\
&= \sum_{i}\sum_{P\in\mathcal{P}_i} \Big(\sum_{e\in P} g_e(f^t_e)\Big)\,(f_P-f^t_P) \\
&= \sum_{i}\sum_{P\in\mathcal{P}_i} \Big(\sum_{e\in P} \xi_e\Big)\,(f_P-f^t_P) \\
&\ge \sum_{i}\sum_{P\in\mathcal{P}_i} \lambda_i\,(f_P-f^t_P) \\
&= \sum_{i} \lambda_i \Big(\sum_{P\in\mathcal{P}_i} f_P - \sum_{P\in\mathcal{P}_i} f^t_P\Big) \\
&= \sum_i \lambda_i (r_i-r_i) \;=\;0,
\end{aligned}
\end{equation*}
where the third equality uses the definition $g_e(f^t_e)=\xi_e$. Hence inequality \eqref{eq:vi} holds.
\end{proof}

\begin{theorem}\label{theo:leq}
    If $\tilde{f}$ and $\tilde{f}'$ are  equilibrium flows for $(G,r,\tilde{\ell},\omega)$, then
    \begin{itemize}
        \item[(a)] $\ell_e(f_e^t)=\ell_e({f_e^t}')$ for all $e \in E$ and 
        \item[(b)] $C(f^t) = C({f^t}')$.
    \end{itemize}
\end{theorem}

\begin{proof}
Let $\tilde{f}$ and $\tilde{f}'$ be two arbitrary equilibria of $(G,r,\tilde{\ell},\omega)$ and let $f^t$ and ${f^t}'$ be their total flow, i.e., $f_e^t=\tilde{f}_{e^R}+\tilde{f}_{e^V}$, ${f_e^t}'=\tilde{f}_{e^R}'+\tilde{f}_{e^V}'$. 
By Lemma~\ref{lem:vi}, both total flows $f^t$ and ${f^t}'$ satisfy the VI. Now for each edge $e \in E$ consider the two single-valued, non-decreasing functions $g_e(\cdot)$ and $g_e'(\cdot)$ as defined above. Note that they could differ only at the critical point $x=2\omega_e/a_e$. By the VI we have 
\begin{equation}\label{eq:VI-1}
\langle g(f^t),\, f - f^t\rangle \ge 0,
\quad \text{for all feasible } f.
\end{equation}
\begin{equation}\label{eq:VI-2}
\langle g'({f^t}'),\, f - {f^t}'\rangle \ge 0,
\quad \text{for all feasible } f.
\end{equation}


Apply Eq. \eqref{eq:VI-1} with $f={f^t}'$ and Eq. \eqref{eq:VI-2} with $f=f^t$:
\begin{equation*}
\langle g(f^t),\, {f^t}' - f^t\rangle \ge 0,
\qquad
\langle g'({f^t}'),\, f^t - {f^t}'\rangle \ge 0.
\end{equation*}
Summing yields
\begin{equation*}
    \sum_{e} [g_e(f_e^t)({f_e^t}'-f_e^t)+g_e'({f_e^t}')(f_e^t-{f_e^t}')] \ge 0,
\end{equation*}
that is
\begin{equation}\label{eq:VIe0}
    \sum_{e}(g_e(f_e^t)-g_e'({f_e^t}'))({f_e^t}'-f_e^t) \ge 0.
\end{equation}

Now consider any edge $e\in E$. When $a_e=0$, we have $g_e(f_e^t)=g_e'({f_e^t}')$ and hence also $\ell_e(f_e^t)=\ell_e({f_e^t}')$. 
When $a_e>0$, note that $g_e(x)$ and $g_e'(x)$ are strictly increasing functions, and they could differ only at the critical point ${2\omega_e}/{a_e}$. Using \eqref{eq:VIe0}, and strict monotonicity of $g_e(\,\cdot\,)$ and $g_e'(\,\cdot\,)$ we claim  that $f_e^t=f_e^{t'}$.

To see why, assume w.l.o.g.\ that $g_e \le g'_e$.
Then $f_e^t<f_e^{t'}$ implies $g_e(f_e^t)\ge g_e'({f_e^t}')$ via \eqref{eq:VIe0}, but strict monotonicity of $g_e(\,\cdot\,)$ demands
$g_e(f_e^t)<g_e(f_e^{t'})\le g_e'(f_e^{t'})$, a contradiction.
On the other hand for $f_e^t>f_e^{t'}$, \eqref{eq:VIe0} implies $g_e(f_e^t)\le g_e'(f_e^{t'})< g_e'(f_e^{t})$, where the last inequality uses strict monotonicity of $g_e'(\,\cdot\,)$. Therefore $g_e(f_e^t)\neq g_e'(f_e^t)$, which implies  $f_e^t=\frac{2\omega_e}{a_e}$, and since $f_e^{t'}<f_e^{t}$, both functions are identical at $f_e^{t'}$, and we 
get the contradiction $g_e'(f_e^{t'})=g_e(f_e^{t'})<g_e(f_e^t)\le g_e'(f_e^{t'})$.
So indeed ${f_e^t}'= f_e^t$ whenever $a_e >0$,
which implies $\ell_e(f_e^t)=\ell_e({f_e^t}')$.





The proof of part (b) uses standard arguments employing the variational inequality from Lemma \ref{lem:vi}:
Let $E_0:=\{e \in E \mid a_e=0 \}$ and $E_+:= E\setminus E_0$. From the arguments above, we have $f_e^t={f_e^t}'$ for $e \in E_+$. Therefore, $$C(f^t)-C({f'}^t) = \sum_{e \in E_0} b_e f^t_e-  \sum_{e \in E_0} b_e  {f_e^t}' .$$ Assume for the sake of contradiction that $\sum_{e \in E_0} b_e f^t_e - \sum_{e \in E_0} b_e {f_e^t}'   > 0$. Then
\begin{align*}
\langle g(f^t), {f^t}'-f^t\rangle
&= \sum_{e\in E} g_e(f^t_e)\,({f_e^t}'-f^t_e) \\
&= \sum_{e\in E_0} g_e(f^t_e)\,({f_e^t}'-f^t_e) \\
&= \sum_{e\in E_0} b_e\,({f_e^t}'-f^t_e) < 0, 
\end{align*}
which contradicts the Lemma \ref{lem:vi} which completes the proof of part~(b).
\end{proof}

\begin{corollary}\label{theo:unique}
     For an instance ($G, r, \tilde{\ell},\omega$) with strictly increasing linear cost functions every equilibrium $\tilde{f}$ has the same total flow $f^t$, where $f_e^t=\tilde{f}_{e^R}+\tilde{f}_{e^V}, \forall e \in E$.  
\end{corollary}




With the equilibrium structure fully characterized, we now turn to the question of optimal pricing and its welfare implications.

\section{Optimal pricing yields socially optimal outcomes }

While equilibrium existence and uniqueness are structural properties of the model, a central question in mechanism design is how to guide the system towards socially optimal outcomes. In classical congestion games, taxes based on marginal cost pricing are known to induce optimal flows, e.g.,  \cite{smith1979b,cole2006much}. We show that a similar principle applies in our model: if each edge charges a priority fee equal to its marginal externality cost, then the resulting equilibrium reproduces the socially optimal flow. This generalizes the marginal cost pricing principle to our priority model, and is the main result of this paper, namely that \emph{voluntary} upgrade options --under the right pricing-- can fully internalize congestion externalities.

\begin{theorem}
 Let $(G,r,\ell)$ be an instance with linear and non\--decreasing latency functions and let $f^*$ be the minimum-cost flow. 
    If $\omega_e=f_e^*\cdot \ell_e'(f_e^*)$ for all $e \in E$ then any equilibrium flow $\tilde{f}$ of the extended instance $(G,r,\tilde{\ell}, \omega)$ realizes optimal total flows, as $\tilde{f}_{e^R}+\tilde{f}_{e^V}=f_e^*\  (\forall e \in E)$, hence the price of anarchy PoA $=1$.
\end{theorem}\label{prop:opt-ns}

\begin{proof}
Recall that the total cost of a given flow $f$ is $C(f)=\sum_{e\in E}f_e\ell_e(f_e)$ $=$ $\sum_{e\in E}f_e(\frac{1}{2}a_ef_e+b_e)$.
Therefore, to enforce socially optimal outcomes, we set the priority prices on each edge $e$ to $\omega_e=f_e^*\cdot \ell_e'(f_e^*)=\frac{1}{2}a_ef_e^*$. We define an extended flow $\tilde{f}$ by assigning the total flow of $f^*$ to the priority option, i.e., for each $e\in E$ we let $\tilde{f}_{e^V} = f^*$ and $\tilde{f}_{e^R}=0$.
We show that $\tilde{f}$ is an equilibrium flow. 



First, by choice of $\tilde{f}$ and prices $\omega$, we note that users are indifferent  between choosing the regular or priority lane, because using \eqref{eq:LV} and \eqref{eq:LR}, for each edge~$e\in~E$ we have  
\begin{align}\nonumber
    \tilde{\ell}_{e^V}(\tilde{f}_{e^V}) &= \frac12a_e\tilde{f}_{e^V}+b_e+\omega_e\\\nonumber
    &=\frac12a_ef_{e}^*+b_e+\frac12 a_ef_e^*\\\nonumber
    &=a_ef_e^*+b_e\\\label{eq:R_V_equal}
    &=a_e\tilde{f}_{e^V}+b_e\ =\ \tilde{\ell}_{e^R}(\tilde{f}_{e^R},\tilde{f}_{e^V})
\end{align}
Now assume for the sake of contradiction that $\tilde{f}$ is not an equilibrium. 
Then there need to exist two different paths $ P, Q \in \mathcal{P}_i$ with $\tilde{f}_{\tilde{P}}>0$, $\tilde{\ell}_{\tilde{P}}(\tilde{f}) > \tilde{\ell}_{\tilde{Q}}(\tilde{f})$ where $\tilde{P}$ and $\tilde{Q}$ are the paths corresponding to $P$ and $Q$, respectively. Now due to \eqref{eq:R_V_equal}, we may assume that all users choose the priority option on all edges of $\tilde{P}$ and $\tilde{Q}$, and therefore
\begin{equation}\label{eq:contradiction}
    \sum_{e\in P}\tilde{\ell}_{e^V}(\tilde{f}_{e^V}) > \sum_{e\in Q}\tilde{\ell}_{e^V}(\tilde{f}_{e^V})\,.
\end{equation}

However since $f^*$ is an optimal flow, optimality conditions for $f^*$ demand that rerouting an infinitesimal fraction of flow from $f_{P}^*$ to $f_{Q}^*$ cannot decrease total cost. Therefore, 
\begin{equation*}
    \sum_{e\in P}(\ell_e(f_e^*)+f_e^*\cdot \ell_e'(f_e^*)) \le \sum_{e\in Q}(\ell_e(f_e^*)+f_e^*\cdot \ell_e'(f_e^*)).
\end{equation*}
Recalling  $\omega_e=f_e^*\cdot \ell_e'(f_e^*)=\frac{1}{2}a_ef_e^*$, this can be easily seen to be equivalent to
\begin{equation*}
    \sum_{e\in P}\tilde{\ell}_{e^V}(\tilde{f}_{e^V}) \le \sum_{e\in Q}\tilde{\ell}_{e^V}(\tilde{f}_{e^V})\,,
\end{equation*}
contradicting \eqref{eq:contradiction}.
Therefore, the extended flow $\tilde{f}$ is an equilibrium flow of $(G,r,\tilde{\ell},\omega)$. Finally, by Theorem \ref{theo:leq}, \emph{every} equilibrium flow of extended instance $(G,r,\tilde{\ell},\omega)$ must have the same total cost (in terms of $(G,r,\ell)$) with the optimal flow $f^*$.
\end{proof}

\section{The effects of uniform pricing}

In the previous section, we observed that when priority lanes are priced at marginal cost, the equilibrium flow in the extended instance aligns precisely with the optimal flow. However implementing heterogeneous pricing schemes is potentially challenging in practice \cite{cole2003pricing}. 
In light of this, we here restrict our attention to uniform pricing schemes in which $\omega_e = \omega_{e'}$ for all $e,e' \in E$.

\begin{theorem}
    There exist instances $(G,r,\tilde{\ell},\omega)$ with linear cost functions in which optimal uniform pricing with price $\omega$ has a PoA of $4/3 - \epsilon$.
\end{theorem}
\begin{proof}

    

\begin{figure}[t]
    \centering
    \includegraphics[width=0.8\linewidth]{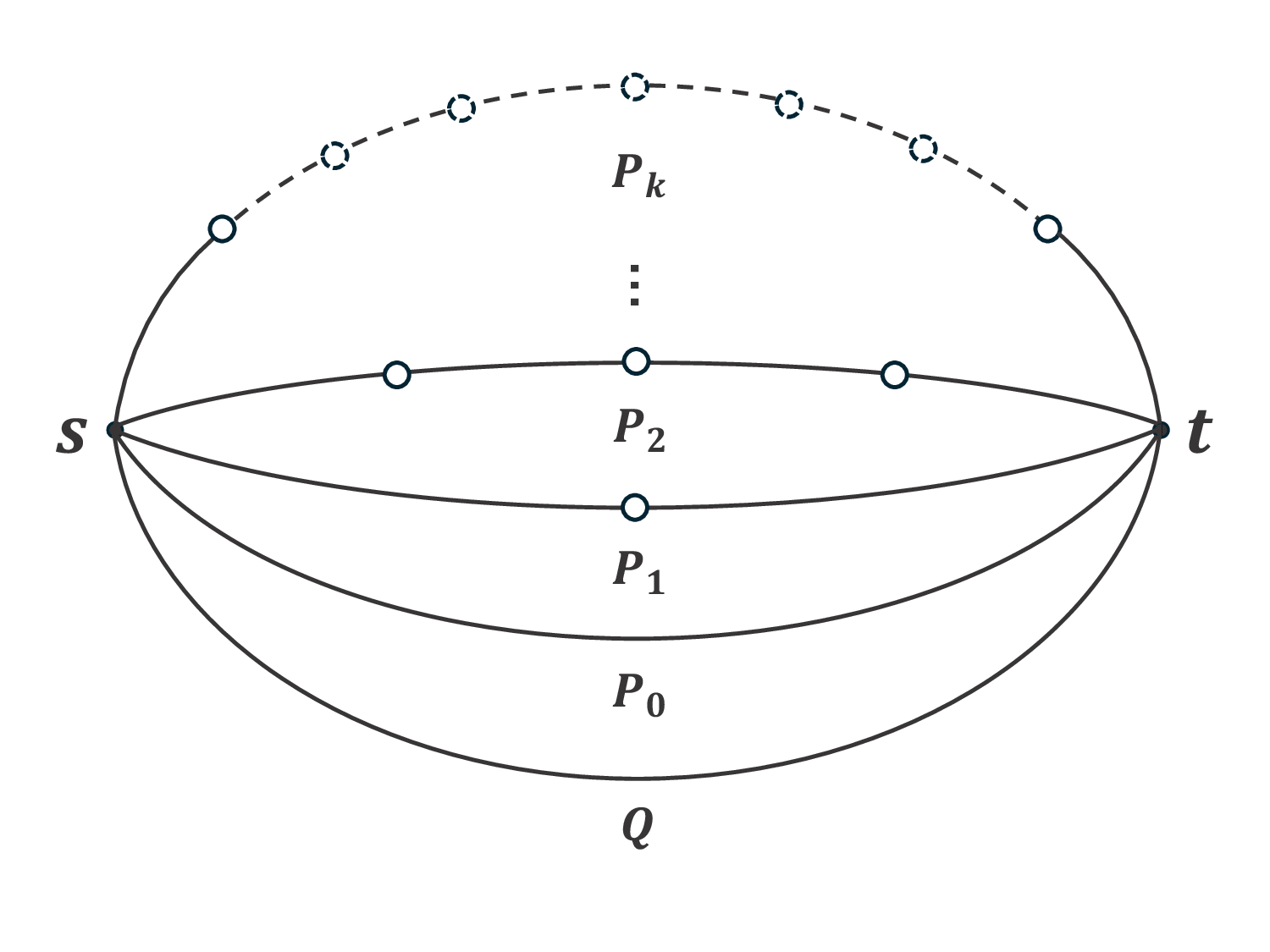}
    \caption{The network for instance  $(G_k,r_k,\tilde{\ell}_k,\omega_k)$  illustrating the structural inefficiency of uniform pricing.}
    \label{fig:4}
\end{figure}

Consider a family of instances parametrized by $k \in \mathbb{N}$ as follows.
The total demand is $k+1$. There are $k+2$ parallel $s$-$t$-paths $Q,P_0,\ldots,P_{k}$. 
The path $Q$ has constant cost of $1$.
Each path $P_i$ consists of a sequence $2^i$ edges where each edge $e \in P_i$ has the cost function  $\hat{\ell}_e(x) = \frac{1}{2^{i-1}} x$. Thus, for each $e \in P_i$ we have that  $\tilde{\ell}_{e^V}(\tilde{f}_{e^V})= \frac{1}{2^{i}} \tilde{f}_{e^V}{+\omega}$ and  $\tilde{\ell}_{e^R}(\tilde{f}_{e^R},\tilde{f}_{e^V}) = \frac{1}{2^{i-1}} \tilde{f}_{e^V} + \frac{1}{2^{i}} \tilde{f}_{e^R}
$.

We will show as $k \rightarrow \infty$ that the $PoA\rightarrow \frac{4}{3}$ for any choice of a uniform price of $\omega$. To that end, we first determine the cost of an optimal flow which routes a flow of $\frac12$ along each path $P_i$ inducing a cost of $\frac{1}{4}$ on each path $P_i$. By routing the remaining demand of $\frac{k+1}{2}$ via $Q$ results in a total cost of $\frac{k+1}{2} \cdot  \frac14 + \frac{k+1}{2} \cdot 1 = (k+1) \frac34 $.

We now consider an equilibrium for a uniform price of $\omega$. We first note, that  in equilibrium the  perceived cost  of any used path  is at most $1$ as path $Q$  has a constant cost of $1$. Furthermore, there is no path in the extended network with perceived cost smaller $1$:
Either (i) the entire demand of $k+1$ is routed over the paths $P_0,\ldots,P_k$ and, therefore, there is a path with flow of at least $1$ and latency (and hence perceived cost) of at least $1$, or (ii) there is flow on  path $Q$ with latency $1$. Hence in equilibrium all used paths must have perceived cost of exactly $1$.
The rest of the analysis proceeds by showing one equilibrium which  yields a perceived cost equal to 1 for all users;
making use of Observation \ref{observation1}. 
\begin{enumerate}
 \item On paths $P_i$ where $\omega < \frac{1}{2^{i+1}}$, there is a flow of $1-2^i \omega$ where every user chooses the priority lanes: 
Similar to Example \ref{ex:1}  (2)  every user chooses priority lanes in equilibrium since for a flow of $r=1-2^i \omega $ we have on each edge $e$ that 
\begin{align*} 
\omega < \frac{1}{2^i} -\omega  =\frac{1}{2} \frac{1}{2^{i-1}} (1-2^i \omega)  = \frac{1}{2}a_e r
\end{align*} which is true for 
$\omega < \frac{1}{2^{i+1}}$.
Thus the perceived cost indeed equals $1$ if the flow equals exactly
$1-2^i \omega $: 
\begin{align*}  
\tilde{\ell}_{P_i} &= \sum_{e \in P_i} \tilde{\ell}_{e^V}(1-2^i \omega) 
=  \sum_{e \in P_i} \left( \frac{1}{2^{i}} (1-2^i \omega) + \omega \right)\\
&= 2^i \cdot \left(\frac{1}{2^i}  (1-2^i \omega )+ \omega \right) = 1 . 
\end{align*} 
\item For paths $P_i$ where $\omega \in \left[\frac{1}{2^{i+1}},\frac{1}{2^i}\right] $, an equilibrium flow 
is attained at the critical value, i.e., when $ \omega=\frac{1}{2^i} r$. This flow of $r=2^i\omega$ splits among priority and regular in such a way that each user's perceived cost equals $1$, as per Observation~\ref{observation1}. 
\item For paths $P_i$ where $\omega > \frac{1}{2^i}$, the price for the priority option is too expensive as $\omega \ge \frac{1}{2} \frac{1}{2^{i-1}} 1$. Again, similar to Example \ref{ex:1} (1) in equilibrium there is a flow of $1$ of regular users with a cost of~$1$.

 \item The remaining demand is routed on path $Q$ with a cost of $1$ as well.
\end{enumerate}
Since $\omega < \frac{1}{2^{i+1}}$ is equivalent to $ i + 1< \log_2 \frac{1}{\omega}$ we denote $ i^* := \lfloor \log_2 \frac{1}{\omega}\rfloor -1$.
Let $I = \left\{i \mid  \omega \in [\frac{1}{2^{i+1}},\frac{1}{2^i}] \right\} $ which is the set of the the one or two indices of paths for case (2) above.
Then, we can express the total latency of the equilibrium flow as a function of $\omega$ as follows: 
\begin{align*}  
c(\omega) = & 
  \sum_{i=0}^{i^*} (1-2^i \omega)^2 +\sum_{i \in I}(2^{i}\omega)^2 \\
  &+ \left((k+1)-\sum_{i \in I}(2^{i}\omega) - (\sum_{i=0}^{i^*} (1-2^i \omega) \right) \cdot 1  
\end{align*}
where the first sum corresponds to the cost of the paths that have priority users only and second term for users at the critical point and the third term accounts for the remaining flow that uses path $Q$ at a cost and latency of $1$. By expanding we get $c(\omega) =$
$$ \sum_{i=0}^{i^*} (1-2^{i+1} \omega + 2^{2i} \omega^2) + \sum_{i \in I}(2^{i}\omega)^2 + (k+1)-\sum_{i \in I} (2^{i}\omega) -\sum_{i=0}^{i^*} (1-2^i \omega)  $$ which allow us to cancel terms and obtain
$$ c(\omega) = \sum_{i=0}^{i^*} (-2^{i} \omega + 2^{2i} \omega^2) +  k+1 + \sum_{i \in I}  (2^{i}\omega)^2 - \sum_{i \in I} (2^{i}\omega).$$
By rewriting the first geometric sum and omitting some positive terms and  bounding the last  sum with $-2$ we can lower bound
$$ c(\omega) \ge -\omega \left(2^{i^*+1} - 1\right)  +  k-1 $$
 and substitute  $ i^* $ with $ \log_2 \frac{1}{\omega} -1 \ge  \lfloor \log_2 \frac{1}{\omega}\rfloor -1$ to further bound 

$$ c(\omega) \ge  -\omega \left(\frac{1}{\omega} - 1\right) +  k-1 = - 1 + \omega +  k   - 1 \ge k -2.$$
Therefore, we conclude that the price of anarchy for this instance is at least ${(k -2)}/{\frac{3(k+1)}{4}} \rightarrow \frac{4}{3}$
as  $k \rightarrow \infty$. 
\end{proof}
As a natural variant one might wonder wether one can obtain improved price of anarchy bounds by allowing the priority option for only a subset of the edges, while still having the restriction of one uniform price for all edges that are equipped with this option. 
However, since voluntary priority options are strategically weaker than mandatory taxation, this additional flexibility does not help in general: Taking a closer look at the previous proof show that this additional flexibility does not improve the total latency in the example above: Restricting the priority option to a subset of the edges only would have an impact on the paths $P_i$ mentioned under (1) and (2) above. However, not allowing the option only (weakly) increases the flow on these paths and, hence, increases total latency. So we conclude:
\begin{corollary}
       There exist instances $(G,r,\tilde{\ell},\omega)$ with linear cost functions in which optimal uniform pricing for any subset of edges has a PoA of $4/3 - \epsilon$.
\end{corollary}

\section{Conclusion}\label{se5}


The key insight of this paper is that properly designed priority pricing options  allows to steer self-interested user behavior to achieve (near-)optimal system performance. In that sense, offering priority options and corresponding optimized prices allows to mimic the well known marginal cost pricing \cite{beckmann1956}, but then as a strategic and optional choice of the users, rather than a mandatory taxation for everybody.  On the other hand, priority lane pricing is strategically weaker: In contrast to mandatory edge costs, it \emph{cannot} (effectively) eliminate edges from the network. Future work could therefore analyze how exactly total perceived cost depends on varying prices. Obvious generalizations are also to consider latency functions other than linear: While many of our proofs would generalize easily, non-linear functions lack the simplicity of the effective cost as illustrated in Figure~\ref{fig:lf}, and fundamentally other effects occur. Finally, while we showed that uniform pricing does generally not suffice to improve social outcomes, it could be true that slightly less restrictive pricing options would suffice to come provably close to socially optimal outcomes.  



\bibliography{mybibfile}

\end{document}